\documentclass[runningheads,a4paper]{llncs}

\usepackage[utf8]{inputenc}

\usepackage{amssymb}

\usepackage{graphicx}   

\usepackage{bold-extra}

\usepackage{booktabs} 

\usepackage{hyperref}

\usepackage{floatrow}
\newfloatcommand{capbtabbox}{table}[][\FBwidth]

\usepackage{url}    
\urldef{\mailsb}\path|{alessandromaria.rizzi,anna.bernasconi}@polimi.it| 
\urldef{\mailsa}\path|claudio.menghi@uni.lu|

\usepackage{paralist}
\usepackage{changepage}
\usepackage{pifont}
\newcommand{\cmark}{\ding{51}}
\newcommand{\xmark}{\ding{55}}
\usepackage{thmtools} 
\usepackage{thm-restate}
\usepackage{mathtools}
\usepackage{mathpartir}
\usepackage[inline]{enumitem}
\usepackage{cancel}
\usepackage{algorithm}
\usepackage{algpseudocode}
\usepackage{multicol}
\usepackage{amsmath}
\usepackage{textgreek}
\usepackage{enumitem}
\usepackage{pifont}
\usepackage{changebar}
\usepackage{multirow}
\usepackage{makecell}

\usepackage{changes}
\definechangesauthor[name={Claudio Menghi},color=orange]{CM}
\definechangesauthor[name={Anna Bernasconi},color=red]{AB}
\definechangesauthor[name={Alessandro Maria Rizzi},color=blue]{AMR}

\usepackage{pgf, pgfplots}
\usepackage{ltl}
\usepackage{tikz}
\usetikzlibrary{arrows,decorations.pathmorphing,backgrounds,positioning,fit,petri}
\usetikzlibrary{automata}
\usetikzlibrary{graphs}
\usetikzlibrary{shadows}
\usetikzlibrary{shapes,arrows}

\usepackage{subcaption}
\usepackage{algorithm}
\newcommand*\circled[1]{\tikz[baseline=(char.base)]{
            \node[shape=circle,draw,inner sep=1pt] (char) {#1};}}

\newcommand{\NAME}{TOrPEDO}

\newcolumntype{P}[1]{>{\centering\arraybackslash}p{#1}}

\usepackage[colorinlistoftodos,prependcaption,textsize=tiny]{todonotes}
\newboolean{showcomments}
\setboolean{showcomments}{true} 
\ifthenelse{\boolean{showcomments}}
{\newcommand{\nb}[2]{
  \fcolorbox{black}{yellow}{\bfseries\sffamily\scriptsize#1}
  {\sf\small$\blacktriangleright$\textit{#2}$\blacktriangleleft$}
 }
 
}
{\newcommand{\nb}[2]{}
 
}

\newcommand\magicfunction{\ensuremath{\texttt{F}}}

\newcommand\tpcompute{\textsc{\texttt{CTP\_KS}}}

\newcommand{\analyze}{\textsc{Analyze}}
\newcommand{\checkalg}{\textsc{Check}}
\newcommand{\tosnf}{\textsc{Sys2Snf}}

\newcommand{\getuc}{\textsc{GetUC}}
\newcommand{\gettp}{\textsc{GetTP}}

\usepackage{booktabs}
\begin{document}

\mainmatter  

\title{Integrating Topological Proofs with Model Checking to Instrument Iterative Design}

\titlerunning{Integrating model checking and topological proofs}

\author{Claudio Menghi\inst{1}
\and Alessandro Maria Rizzi\inst{2}
\and Anna Bernasconi\inst{2}
}

\authorrunning{Integrating Topological Proofs with Model Checking}

\institute{
University of Luxembourg, Luxembourg\\
\mailsa\\
\and
Politecnico di Milano, Italy\\
\mailsb\\
}

\toctitle{Lecture Notes in Computer Science}
\tocauthor{Authors' Instructions}
\maketitle

\begin{abstract}
System development is not a linear, one-shot process. It proceeds through refinements and revisions. 
To support assurance that the system satisfies its requirements, it is desirable that continuous verification can be performed after each refinement or revision step.
To achieve practical adoption, formal system modeling and verification must accommodate continuous verification efficiently and effectively.
Our proposal to address this problem is \NAME , a verification approach where models are given via Partial Kripke Structures (PKSs) and requirements are specified as Linear-time Temporal Logic (LTL) properties.
PKSs support refinement, by deliberately indicating unspecified parts of the model that are later completed.
We support verification in two complementary forms: via model checking and proofs. Model checking is useful to provide counterexamples, i.e., pinpoint model behaviors that violate requirements. Proofs are instead useful since they can explain why requirements are satisfied. In our work, we introduce a specific concept of proof, called topological proof (TP). A TP produces a  slice of the original PKS which justifies the property satisfaction. Because models can be incomplete, \NAME\ supports reasoning on requirements satisfaction, violation, and possible satisfaction (in the case where the satisfaction depends on unknown parts).

\end{abstract}

\section{Introduction}

One of the goals of software engineering and formal methods is to provide automated verification tools that support designers in producing models of an envisioned system, which follows a set of properties of interest.
Many automated verification methods are available to help and guide the designer through this complex task. These methods include, among others, model checking and theorem proving. Typically, the designer benefits from automated support to understand why her system does not behave as expected (e.g., counterexamples), but she might find useful also information retrieved when the system already follows the specified requirements. While model checkers provide the former, theorem provers sustain the latter.
These usually rely on some form of deductive mechanism that, given a set of axioms, iteratively applies a set of rules until a theorem is proved. The proof consists of the specific sequence of deductive rules applied to prove the theorem. In literature, many approaches have dealt with integration of model checking and theorem proving at various levels (e.g.,~\cite{namjoshi2001certifying,cleaveland2002evidence,rajan1995integration,kupferman2005complementation}). These approaches are oriented to provide \textit{certified model checking} procedures rather than tools which actually help the design process.
Even when the idea is to provide a practically useful framework~\cite{PPZ01,PZ01}, the output consists of deductive proofs which are usually difficult to understand and hardly connectable with the designer's modeling choices. 
Moreover, verification techniques which only take into account completely specified designs do not comply with modern agile software design techniques.
In a recent work (\cite{Bernasconi2017,bernasconi2017arxiv}), we have considered cases in which a partial knowledge of the system model is available. However, the presented approach was mainly theoretical and lacked a practical implementation.

With the intent to provide a valuable support for a flexible design process, we formulate the initial problem on models that contain uncertain parts. Partial specification may be used, for instance, to represent the uncertainty of introducing, keeping, excluding particular portions of the design with respect to the complete model of the system. 
For this reason, we  chose Partial Kripke Structures (PKSs) as a formalism to represent general models. 
PKSs are a standard formalism used to reason on incomplete systems.
Among several applications, they have been used in requirement elicitation to reason about the system behavior from different points of view~\cite{easterbrook2001framework,brunet2006manifesto}.
Furthermore, other modeling formalisms such as Modal Transition Systems~\cite{larsen1988modal} (MTSs), can be converted into PKSs through a simple transformation~\cite{godefroid2003expressiveness}. 
Thus, the proposed solution can also be easily applied on models specified using MTSs, which are commonly used in software development~\cite{foster2006ltsa,UchitelFm}.
Kripke Structures (KSs) are particular instances of PKSs used to represent complete models.  Requirements on the model are expressed in Linear-time Temporal Logic (LTL).
Verification techniques  that consider PKSs return three alternative values: \emph{true} if the property holds in the partial model, \emph{false} if it does not hold, and \emph{maybe} if the property possibly holds, i.e., its satisfaction depends on the parts that still need to be refined.

Methods for verifying partial models naturally fit in modern software design processes~\cite{FASE18,foster2006ltsa,UchitelFm}.
In the iterative design process, the designer starts from a high level model of the system in which some portions can be left unspecified, representing design decisions that may be taken in later development steps.
As development proceeds, the initial model can be \emph{refined}, by filling parts that are left unspecified, or \emph{revised}, by changing parts that were already specified.
In a PKS a refinement is performed by associating a true/false value to previously unknown propositions, while revising may also involve adding or removing existing states or transitions, or changing values already assigned to propositions (i.e., a refinement is a type of revision).

A comprehensive integrated design framework able to support software designers in understanding \emph{why} properties are (possibly) satisfied -- as models are specified, refined, or revised -- is still missing.

We tackle this problem by presenting \NAME\ (TOpological Proof drivEn Development framewOrk), a novel automated verification framework, that:
\begin{itemize}
\item[(i)] supports a modeling formalism which allows a partial specification of the system design;
\item[(ii)] allows performing analysis and verification in the context of systems in which ``incompleteness'' represents a conceptual uncertainty;
\item[(iii)] provides guidance in the refinement process through complementary outputs: counterexamples and topological proofs;
\item[(iv)] when the system is completely specified, allows understanding which changes impact or not on certain properties.
\end{itemize} 

\NAME\ is based on the novel notion of \textit{topological proof} (TP), which tries to overcome the complexity of deductive proofs and is designed to make proofs understandable on the original system design.
A TP is a \textit{slice} of the original model that specifies which part of it influences the property satisfaction.
If the slice defined by the TP is not preserved during a refinement or a revision,
there is no assurance that the property holds (possibly holds) in the refined or revised model.
This paper proposes an algorithm to compute topological proofs---which relies on the notion of \textit{unsatisfiable cores} (UCs) \cite{SCHUPPAN2016155}---and proves its correctness on PKSs. 
\NAME\  has been implemented on top of NuSMV~\cite{nusmv} and PLTL-MUP~\cite{sergeantfinding}.
The implementation has been used to evaluate how \NAME\ helps software designers by considering a set of examples coming from literature including both completely specified and partially specified models.

The paper is structured as follows.
Section~\ref{sec:torpedo} describes \NAME.
Section~\ref{sec:background} discusses the background.
Sections~\ref{sec:topologicalproof} and \ref{sec:automatedsupport} present the theoretical results and the algorithms that support \NAME.
Section~\ref{sec:evaluation} evaluates the achieved results.
Section~\ref{sec:related} discusses related work.
Section~\ref{sec:conclusions} concludes.
 \section{\NAME}
\label{sec:torpedo}

\NAME\ is a proof based development framework which allows verifying initial designs and evaluating their revisions.
To illustrate \NAME , we use an example: a designer needs to develop a simple vacuum-cleaner robot which has to satisfy the requirements  in Table~\ref{tab:motivatinproperties}, specified through LTL formulae and plain English text. 
These are simple requirements that will be used for illustration purposes.
Informally, when turned on, the vacuum-cleaner agent can move with the purpose to reach a site which may be cleaned in case it is dirty.

\begin{figure}
\begin{floatrow}\CenterFloatBoxes
\ffigbox{
  \tikzset{every loop/.style={in=120,out=150,looseness=4}}

\begin{adjustwidth}{-2em}{} 
\begin{tikzpicture}[->,>=stealth',shorten >=1pt,auto,node distance=2cm,
                    thick,main node/.style={circle,draw,font=\sffamily\Large\bfseries}]

\node[state,initial,align=center, minimum size=1.2cm,initial text=]          (off) 
[label={
[align=left,yshift=0.1cm]
above:
OFF}]   {$move=\LTLfalse$\\ $suck=\LTLfalse$\\ $on=\LTLfalse$\\ $reached=\LTLfalse$};

\node[state,align=center,  minimum size=1.2cm]         (idle1) 
[
right= 0.6cm of off,    
label={
[align=left,yshift=0.1cm]
above:
IDLE}
]  {$move=\LTLfalse$\\ $suck=\LTLfalse$\\ $on=\LTLtrue$\\ $reached = ?$};
 
\node[state,align=center, minimum size=1.2cm]          (idle2) 
[below =0.7cm  of off,
label={
[align=left,yshift=0.1cm]
above:
CLEANING}]  {$move=?$\\ $suck=\LTLtrue$\\ $on=\LTLtrue$\\ $reached = \LTLtrue$};

\node[state,align=center, minimum size=1.2cm]          (ac) 
[
 right= 0.6cm of idle2,
label={
[align=left,yshift=0.1cm]
above:
MOVING}]  {$move=\LTLtrue$\\ $suck=?$\\ $on=\LTLtrue$\\ $reached=?$};

\path[->]      
(off) edge [bend left=5] node [above]{} (idle1)
(idle1) edge [bend right=5] node [above]{} (off)

(idle2) edge node [above]{} (idle1)

(idle1) edge [bend left=40] node [above]{} (ac)

(ac) edge [bend left=5]node [above]{} (idle2)

(off) edge [loop above] node{} ()
(idle1) edge [loop above] node{} ()
(idle2) edge [loop below] node{} ()
(ac) edge [loop below] node{} ();

\end{tikzpicture}
\end{adjustwidth} }{
  \caption{Model of a vacuum-cleaner robot.}
    \label{fig:motivatingmodel}
}
\ttabbox{

\begin{tabular}{l}
\toprule
	\textbf{LTL formulae} \\
	\midrule
$\phi_1=\LTLg (\mathit{suck} \rightarrow \mathit{reached})$ \\
$\phi_2=\LTLg ((\neg \mathit{move}) \LTLw \mathit{on} )$ \\
$\phi_3=\LTLg(((\neg \mathit{move}) \wedge \mathit{on}) \rightarrow \mathit{suck})$ \\		
$\phi_4=((\neg \mathit{suck}) \LTLw (\mathit{move} \land (\neg \mathit{suck}) ))$ \\ 
		\midrule
		\textbf{Textual requirements} \\ 
	\midrule
	\begin{minipage}{2in}
   \vskip 4pt
	$\phi_1$: the robot is drawing dust ($\mathit{suck}$) only if it $\mathit{reached}$ the cleaning site.
	
	$\phi_2$: the robot must be turned $\mathit{on}$ before it can $\mathit{move}$. 
	
	$\phi_3$: if the robot is $\mathit{on}$ and stationary ($\neg \mathit{move}$), it must be drawing dust ($\mathit{suck}$).
	
	$\phi_4$: the robot must $\mathit{move}$ before it is allowed to draw dust ($\mathit{suck}$).
    \vskip 4pt
	\end{minipage}
	\\ 
	\bottomrule
\end{tabular}	
 }{
  \caption{Sample requirements.}
  \label{tab:motivatinproperties}
}
\end{floatrow}
\vspace{-0.3cm}
\end{figure}

The \NAME\ framework is illustrated in Fig.~\ref{fig:renovated} and it is made of four phases: \textsc{initial design}, \textsc{analysis}, \textsc{revision}, and \textsc{re-check}. 
Dashed boxes marked with the person icon represent phases performed manually by the designer, while dashed boxes marked with the gears icon contain phases operated using automated support.

\textbf{\textsc{Initial design.}} This phase concerns the initial definition of the model of the system, formalized as a PKS (marked in Fig.~\ref{fig:renovated} with \circled{\scriptsize{1}}) along with the property of interest, in LTL (\circled{\scriptsize{2}}).

In the vacuum-cleaner example, the designer has identified two actions the robot can perform:
$\mathit{move}$, i.e., the agent travels to the cleaning site;
$\mathit{suck}$, i.e., the agent is drawing the dust.
She has also identified two conditions that can trigger actions:
$\mathit{on}$, true when the robot is turned on;	
$\mathit{reached}$, true when the robot has reached the cleaning site.   
These actions and conditions determine the designer description of the preliminary model presented in Fig.~\ref{tab:motivatinproperties}.
The model is made by four states representing the configuration of the vacuum-cleaner robot.
The state $\mathit{OFF}$ represents the robot being shut down, 
$\mathit{IDLE}$ the robot being tuned in w.r.t. a cleaning call, 
$\mathit{MOVING}$ the robot reaching the cleaning site, 
and $\mathit{CLEANING}$ the robot performing its duty.

Each state is labeled with the actions and conditions that are true in that state. 
Given an action or condition $\alpha$ and a state $s$, we use the notation:  
$\alpha=\LTLtrue$ to indicate that $\alpha$ occurs when the robot is in state $s$;
$\alpha=\LTLfalse$ to indicate that $\alpha$ does not occur when the robot is in state $s$;
$\alpha=?$ to indicate that there is uncertainty on whether $\alpha$ occurs when the robot is in state $s$.
In the first two cases the designer is sure that an action must (must not) be performed or a condition must be true (false) in a state; in the third case the designer is uncertain about the design.
Specifically, she does not know whether the robot should perform an action or whether a condition should be true in order for a state to be entered.

\begin{figure*}[t]     
\includegraphics[width=\linewidth]{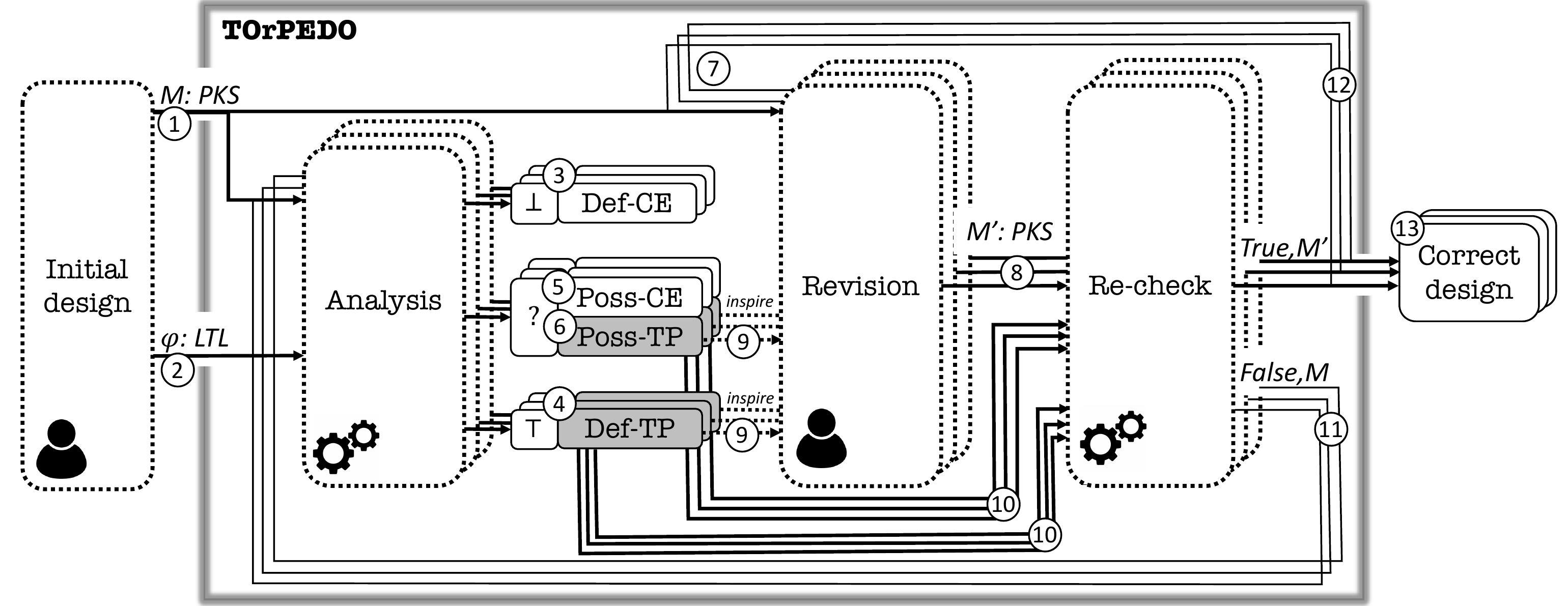}
	\caption{\NAME\ structure. Continuous arrows represent inputs and outputs to phases. Circled numbers are used to reference the image in the text.}  
\label{fig:renovated}
\vspace{-0.3cm}
\end{figure*}

\textbf{\textsc{Analysis.}} \NAME\ assists the designer with an automated analysis, which includes the following elements:

\begin{enumerate}
\item[(i)] information about \emph{what is wrong} in the current design. This information includes a definitive-counterexample, which indicates a behavior that depends on already performed design choices and violates the properties of interest. The definitive-counterexample (i.e., Def-CE \circled{\scriptsize{3}}) can be used to improve $M$ and produce a revised version $M^\prime$ that satisfies or possibly satisfies the property of interest.
\item[(ii)] information about \emph{what is correct} in the current design. This information includes definitive-topological proofs (i.e., Def-TP \circled{\scriptsize{4}}) that indicate a portion of the design that ensures property satisfaction;
\item[(iii)] information about \emph{what could be wrong}/\emph{correct} in the current design depending on how uncertainty is removed. This information includes: a possible-counterexample (i.e., Poss-CE \circled{\scriptsize{5}}), indicating a behavior (which depends on uncertain actions and conditions) that violates the properties of interest, and a possible-topological proof (i.e., Poss-TP \circled{\scriptsize{6}}), indicating a portion of the design that ensures the possible satisfaction of the property of interest. The designer can use the possible-counterexample and the possible-topological proof to improve $M$.
\end{enumerate}

In the following we will use the notation $x$-topological proofs or $x$-TP to indicate arbitrarily definitive-topological or possible-topological proofs.

In the vacuum-cleaner example, the designer analyzes her proposed design using \NAME .

\begin{sloppypar}
Property $\phi_1$ is possibly satisfied. \NAME\ returns the possible-counterexample $\mathit{OFF}$, $\mathit{IDLE}$, $(\mathit{MOVING})^{\omega}$.
This possible-counterexample shows a run that violates the property of interest.
\NAME\ also returns a possible-topological proof showing that the property remains possibly satisfied given that $\mathit{OFF}$ remains the only initial state, $reached$ still holds in $\mathit{CLEANING}$, and $suck$ does not hold in $\mathit{OFF}$ and $\mathit{IDLE}$, 
while unknown in $\mathit{MOVING}$ (note that, if $suck$ was set to $\bot$ in this state we would indeed obtain a proof). In addition, all transitions must be preserved.
\end{sloppypar}

Property $\phi_2$ is satisfied. \NAME\ returns the definitive-topological proof, which shows that the property remains satisfied given that $\mathit{OFF}$ remains the only initial state, $on$ still holds in $\mathit{MOVING}$ and $\mathit{CLEANING}$, and $move$ does not hold in $\mathit{OFF}$ and $\mathit{IDLE}$. In addition, all transitions must be preserved.

Property $\phi_3$ is not satisfied. \NAME\ returns a definitive-counterexample, e.g., $\mathit{OFF},IDLE^\omega$. 
The counterexample shows that it is not true that always a robot which is operative and not moving is drawing dust.

Property $\phi_4$ is possibly satisfied. \NAME\ returns the possible-counter- example $\mathit{OFF}$, $(\mathit{IDLE},\,\mathit{MOVING},\,\mathit{CLEANING},\,\mathit{IDLE},\,\mathit{OFF})^{\omega}$, which specifies a sample run for which it is not true that the robot is only moving (and not cleaning) before it can draw dust. 
The topological proof shows that the property remains possibly satisfied given that the following characteristics of the model are preserved: 
from the only initial state $\mathit{OFF}$ one can loop or move to $\mathit{IDLE}$, from which one can loop, return to $\mathit{OFF}$, or go to $\mathit{MOVING}$; in addition $move$ must hold in $\mathit{MOVING}$ and $suck$ must not occur in $\mathit{OFF}$ and $\mathit{IDLE}$, while unknown in $\mathit{MOVING}$.

\textbf{\textsc{Revision.}} As development proceeds, the designer may want to revise the existing model by 
changing some of its parts: adding/removing states and transitions or by changing propositions labelling inside states.
Revision may include refinement, i.e., replacing with $\LTLtrue$ and $\LTLfalse$ some unknown values in the atomic propositions. 
The inputs of this phase are the initial model $M$ (\circled{\scriptsize{1}}), or an already revised model (\circled{\scriptsize{7}}), 
and the $x$-TP that can be used by the designer as a guideline for the revision (\circled{\scriptsize{9}}).
 The output is another revised model $M'$ (\circled{\scriptsize{8}}).

\textit{Revision 1.} The designer would like her model to not violate any property of interest.
She examines the counterexample of $\phi_3$ to understand why it is not satisfied and envisions a revised model that could satisfy the property.
She also consults the $x$-TPs of properties $\phi_1$, $\phi_2$, and $\phi_4$ in order to be sure to preserve their verification results.
She thus decides to change the value of $move$ in state $\mathit{IDLE}$ from $\bot$ to $\top$.
Since she foresees $\phi_3$ is now satisfied, she reruns the \textsc{analysis} for this property.
\NAME\ provides the corresponding $x$-TP.

\textit{Revision 2.} The designer decides to further improve her model by proposing a refinement: $move$ becomes $\top$ in state $\mathit{CLEANING}$ and $reached$ becomes $\bot$ in state $\mathit{IDLE}$.
Since $\phi_1$, $\phi_2$, $\phi_3$, and $\phi_4$ were previously not violated, \NAME\ performs the \textsc{re-check} phase for each property.

\textbf{\textsc{Re-check.}} The automated verification tool provided by \NAME\ checks whether all the changes in the current model revision (\circled{\scriptsize{8}}) are compliant with the $x$-TPs (\circled{\scriptsize{10}}), i.e., changes applied to the revised model do not include parts that had to be preserved according to the proof.
If a property of interest is (possibly) satisfied in a previous model (\circled{\scriptsize{1}}), and the revision of the model is compliant with the property $x$-TP, the designer has the guarantee that the property is (possibly) satisfied in the revision.
Thus, she can perform another model revision round (\circled{\scriptsize{12}}) or approve the current design (\circled{\scriptsize{13}}).
Otherwise, \NAME\ re-executes the \textsc{analysis} (\circled{\scriptsize{11}}). 

In the vacuum-cleaner case, the second revision passes the \textsc{re-check} and the designer proceeds to a new revision phase.

 \section{Background}
\label{sec:background}
We present some background and notation necessary to understand the rest of the paper.
First, we describe how to model the system under development and its properties. 
Then, we present the unsatisfiable core which is the element upon which our algorithm for  computing topological proofs is based. 

\subsection{Modeling  systems and  properties}
\label{sub:modeling}
We first describe Partial Kripke Structures (PKS)
and Kripke Structures (KS), formalisms that allow modeling the systems under development. 
Then, we briefly introduce the semantics for LTL properties on PKSs and KSs and how to perform model checking on these structures. 

PKS are a modeling formalism that can be adopted when the value of some propositions is uncertain on selected states. 
\begin{definition}[\cite{bruns1999model}] 
	A \emph{Partial Kripke Structure} $M$ is a tuple $\langle S, R,S_0,AP,$ $L \rangle$, where:
$S$ is a set of states;
$R\subseteq S\times S$ is a left-total transition relation on $S$;
$S_0$ is a set of initial states;
$AP$ is a set of atomic propositions;
$L: S\times AP\rightarrow \{\top,?,\bot\}$ is a function that, for each state in $S$, associates a truth value to every atomic proposition in $AP$.
\end{definition}
Informally, a PKS represents a system as a set of states and transitions between these states. 
Uncertainty on the $AP$ is represented through the value $?$.
The model of the vacuum-cleaner agent presented in Fig.~\ref{fig:motivatingmodel} is a PKS where propositions in $AP$ are used to model actions and conditions.

\begin{definition}[\cite{Kripke1963-KRISCO}] 
	A \emph{Kripke Structure}  $M$ is a PKS $\langle S, R,S_0,AP,L \rangle$, where $L: S\times AP\rightarrow \{\top,\bot\}$.
\end{definition}

PKSs can be related to other PKSs or to KSs respectively through \textit{refinement} and \textit{completion}.

\begin{definition}[\cite{bruns2000generalized}] 
\label{def:refinement}
 Let $M=\langle S, R,S_0,AP,L \rangle$ be a PKS.
 A refinement of $M$ is a PKS $M^\prime=\langle S, R, S_0,$ $AP,$ $L^\prime \rangle$ 
 where $L^\prime$ is such that
\begin{itemize}
\item  for all $s\in S$, $\alpha \in AP$ if $L(s, \alpha)=\LTLtrue \rightarrow L^\prime(s, \alpha)=\LTLtrue$;
\item  for all $s\in S$, $\alpha \in AP$ if $L(s, \alpha)=\LTLfalse \rightarrow L^\prime(s, \alpha)=\LTLfalse$.
\end{itemize} 
\end{definition}

We indicate that $M^\prime$ is a refinement of $M$ using the notation  $M  \preceq M^\prime$. 
Intuitively, the notion of refinement allows assigning a $\LTLtrue$ or a $\LTLfalse$ value to an atomic proposition $\alpha$ in a state $s$ s.t. $L^\prime(s, \alpha)=?$.

\begin{definition}[\cite{bruns2000generalized}] 
Let $M$ be a PKS and let $M^\prime$ be a KS. 
Then $M^\prime$ is a \emph{completion} of $M$ if and only if  $M \preceq M^\prime$.
\end{definition} 

Intuitively, a completion of a PKS is a KS obtained by assigning a $\LTLtrue$ or a $\LTLfalse$ value to every atomic propositions $\alpha$ and state $s$ s.t. $L^\prime(s, \alpha)=?$.

\vspace{2mm}
\noindent
\textbf{Semantics of LTL properties.} 
For KSs we consider the classical LTL semantics $[M \models \phi]$ over infinite words that associates to a model $M$ and a formula $\phi$ a truth value in the set $\{ \bot, \top \}$. The interested reader may refer, for example, to~\cite{katoen2008}.
Let $M$ be a KS and $\phi$ be an LTL property. We assume that the function $\checkalg$, such that $\langle res, c \rangle =\checkalg (M $, $\phi)$, returns a tuple $\langle res, c \rangle$, where $res$ is the model checking result in $\{ \top, \bot \}$ and, if $res=\bot$, $c$ is the counterexample.

Instead, when the satisfaction of  LTL over PKSs is of interest, two semantics can be considered: the three-valued or the thorough semantics.

The \emph{three-valued LTL semantics}~\cite{bruns1999model}  $[M \models_{3} \phi]$ associates to a model $M$ and a formula $\phi$ a truth value in the set $\{ \bot, ?, \top \}$ and is defined based on the information ordering $\top >\, ? > \bot$, i.e., on the assumption that $\top$ ``provides more information'' than $?$ and $?$ ``provides more information'' than $\bot$~\cite{bruns1999model}.
The three-valued LTL semantics is defined by considering paths of the model $M$. A path $\pi$ is a sequence of states $s_0,s_1,\ldots$ such that, for all $i \geq 0$, $(s_i, s_{i+1}) \in R$.

\begin{definition}[\cite{bruns1999model}]
Let $M = \langle S, R,$ $S_0, AP, L \rangle$ be a PKS,
let $\pi=s_0,s_1,\ldots$ be a path, and 
let $\phi$ be an LTL formula. Then, the \emph{three-valued semantics} $[(M,\pi)\models_3\phi]$ is defined inductively as follows:
\begin{align*}
&[(M,\pi) \models_3 p] & = &&& L(s_0,p)\\
&[(M,\pi) \models_3 \lnot\phi] & =  &&& comp([(M,\pi) \models_3 \phi]) \\
&[(M,\pi) \models_3 \phi_1 \LTLand \phi_2] & = &&& \min([(M,\pi) \models_3 \phi_1],[(M,\pi) \models_3 \phi_2])\\
&[(M,\pi) \models_3 \LTLnext \phi] & =  &&& [(M,\pi^1) \models_3 \phi]\\
&[(M,\pi) \models_3 \phi_1 \LTLuntil \phi_2] & = &&& \max_{j\geq 0}(\min(\{[(M,\pi^i) \models_3 \phi_1]|i<j\} \cup \{[(M,\pi^j) \models_3 \phi_2]\}))
\end{align*}
\end{definition}

The conjunction (resp. disjunction) is defined as the minimum (resp. maximum) of its arguments, following the order $\bot<\,?<\top$. These functions are extended to sets with min($\emptyset$)=$\top$ and max($\emptyset$)=$\bot$.  The $comp$ operator maps $\top$ to $\bot$, $\bot$ to $\top$, and $?$ to $?$

\begin{definition}[\cite{bruns1999model}]
Let $M = \langle S, R,$ $S_0, AP, L \rangle$ be a PKS,
$s$ be a state of $M$
and $\phi$ be an LTL formula. Then
$ [(M, s) \models_3 \phi]=\min(\{[(M, \pi) \models_3 \phi] \mid \pi^0=s\})$.
\end{definition}

Intuitively this means that, given a formula $\phi$, each state $s$ of $M$ is associated with the minimum of the values obtained considering the LTL semantics over any path $\pi$ that starts in $s$.

When the three-valued semantics is considered, if $[M \models_{3} \phi] = \top$ then in every completion of $M$ formula $\phi$ is true and if $[M \models_{3} \phi] = \bot$, then in every completion of $M$ formula $\phi$ is false.
In general, when $[M \models_{3} \phi] = ?$, there exist both completions of $M$ that satisfy and do not satisfy $\phi$.
However, there are cases in which all the completions of $M$ satisfy (or do not satisfy) $\phi$.
For this reason, the alternative thorough LTL semantics~\cite{bruns2000generalized} has been proposed ($[M \models_{T} \phi]$). 

The \emph{thorough LTL semantics}~\cite{bruns2000generalized} dictates that $[M \models_{T} \phi] = \top$ if  in every completion of $M$ formula $\phi$ is true and  $[M \models_{T} \phi] = \bot$ if in every completion of $M$ formula $\phi$ is false.
This ensures that, when $[M \models_{T} \phi] =?$, there exist both completions of $M$ that satisfy $\phi$ and completions of $M$ that do not satisfy $\phi$.

\begin{definition}[\cite{bruns2000generalized}]
Let $M$ be a PKS and 
let $\phi$ be an LTL formula.
Then,
\begin{equation}
    [M \models_T \phi] \buildrel \text{def}\over = 
                						 \begin{cases}
                  				 				\top & \quad \text{if}\ M^\prime \models \phi \text{ for every completion } M^\prime \text{ of } M\\
                  								\bot & \quad \text{if}\ M^\prime \not\models \phi \text{ for every completion } M^\prime \text{ of } M \\
                  								? & \quad  \text{otherwise}
                						\end{cases}
\end{equation}
\end{definition}
Note that, when a PKS is a KS, $[M \models_{3} \phi]= [M \models_{T} \phi]=[M \models \phi]$.

\begin{lemma}[\cite{godefroid2011ltl}]
Let $M$ be a PKS and 
let $\phi$ be an LTL formula.
Then 
\begin{enumerate*}
\item[] $[M \models_3 \phi] =$ $ \top \Rightarrow$ $ [M\models_T \phi] = \top$; and
\item[] $[M\models_3\phi] = \bot$ $\Rightarrow$ $[M\models_T \phi] = \bot$.
\end{enumerate*}
\end{lemma}
That is, a formula which is true (false) under the three-valued semantics is also true (false) under the thorough semantics.

There exists a subset of LTL formulae, known in the literature as \emph{self-minimizing}~\cite{godefroid2005model}, such that the two semantics coincide, i.e.,
given a model $M$ and a  self-minimizing LTL property $\phi$, then $[M\models_3\phi]=[M\models_T\phi]$. It has been observed that most practically useful LTL formulae belong to this subset~\cite{godefroid2005model}.

\vspace{2mm}
\noindent
\textbf{Model checking.}
Checking KSs with respect to LTL properties can be done by using classical model checking procedures.
For example, the model checking problem of property $\phi$ on a KS $M$ can be reduced to the satisfiability problem of the LTL formula $\Phi_M \LTLand \neg \phi$, where $\Phi_M$ represents the behaviors of model $M$.
If $\Phi_M \LTLand \neg \phi$ is satisfiable, then $[M \models \phi]=\bot$, otherwise $[M \models \phi]=\top$.

Checking a PKS $M$ with respect to an LTL property $\phi$ considering the three-valued semantics can be done by performing twice the classical model checking procedure for KSs~\cite{bruns2000generalized}, one considering an optimistic approximation $M_{opt}$ and one considering a pessimistic approximation $M_{pes}$.
These two procedures consider the LTL formula $\phi^\prime=\magicfunction(\phi)$, where \magicfunction\ transforms $\phi$ with the following steps:  
\begin{enumerate}
\item[(i)] negate $\phi$; 
\item[(ii)] convert  $\neg \phi$ in negation normal form\footnote{An LTL formula $\phi$ is in \emph{negation normal form} if negations are applied only to atomic propositions.
Conversion of an LTL formula into its negation normal form can be achieved by pushing negations inward and replacing them with their duals---for details see \cite{katoen2008}.};
\item[(iii)] replace every subformula $\neg \alpha$, where $\alpha$ is an atomic proposition, with a new atomic proposition $\overline{\alpha}$.
\end{enumerate}

To create the  optimistic  and pessimistic  approximations  $M_{opt}$ and  $M_{pes}$,
the PKS $M=\langle S, R,S_0,AP,L \rangle$ is first converted into its \emph{complement-closed} version $M_c=\langle S, R, S_0, AP_c, L_c \rangle$ where the set of atomic propositions $AP_c = AP \cup \overline{AP}$
is such that $\overline{AP}=\{\overline{\alpha} \mid \alpha \in AP \}$.
Atomic propositions in $\overline{AP}$ are  called complement-closed propositions.
Function $L_c$ is such that for all $s \in S$ and $\alpha \in AP$,  $L_c(s, \alpha)=L(s,\alpha)$
and for all $s\in S$ and $\overline{\alpha} \in \overline{AP}$, $L_c(s,\overline{p})=comp(L(s,p))$.
For example, the complement-closed PKS of the vacuum-cleaner agent in Fig.~\ref{fig:motivatingmodel}, in the state $IDLE$ presents eight propositional assignments: $move=	\bot$, $\overline{move}=\top$, $suck=\bot$, $\overline{suck}=\top$, $on=\top$, $\overline{on}=\bot$, $reached=?$, and $\overline{reached}=?$.

The two model checking runs for a PKS $M=\langle S, R,S_0,AP,L \rangle$ are based respectively on an optimistic ($M_{opt}=\langle S, R,S_0,AP_c,L_{opt} \rangle$) and a pessimistic ($M_{pes}=\langle S, R,S_0,AP_c,L_{pes} \rangle$) approximation of $M$'s relative complement-closed $M_c=\langle S, R, S_0, AP_c, L_c \rangle$.
Function $L_{pes}$ (resp. $L_{opt}$) is such that for all $s \in S$, $\alpha \in AP_c$, and $L_c(s, \alpha) \in \{\top, \bot \}$, then $L_{pes}(s, \alpha)=L_c(s,\alpha)$ (resp. $L_{opt}(s, \alpha)=L_c(s,\alpha)$),
and for all $s\in S$,  $\alpha \in AP_c$, and $L_c(s, \alpha)=?$, then $L_{pes}(s,\alpha)=\bot$ (resp. $L_{opt}(s,\alpha)=\top$).

Let $A$ be a KS and $\phi$ be an LTL formula, $A$ $\models^\ast\phi$ is true if no path that satisfies the formula $\magicfunction (\phi)$ is present in $A$.

\begin{theorem}[\cite{bruns1999model}]
\label{th:threevaluedMC}
Let $\phi$ be an LTL formula, 
let $M=\langle S, R, S_0, AP, L\rangle $ be a PKS,
and let $M_{pes}$ and $M_{opt}$ be the pessimistic and optimistic approximations of $M$'s relative complement-closed $M_c$. Then
\begin{equation}
    [M \models_3\phi]\buildrel \text{def}\over = 
                						 \begin{cases}
                  				 				\top & \quad \text{if}\ M_{pes}\models^\ast\phi\\
                  								\bot & \quad \text{if}\ M_{opt}\not\models^\ast\phi\\
                  								? & \quad  otherwise
                						\end{cases}
\end{equation}
\end{theorem}

We assume that the function $\checkalg^\ast$ computes the result of operator $\models^\ast$. It takes as input either $M_{pes}$ or $M_{opt}$ and the property $\magicfunction(\phi)$, and returns a tuple $\langle res, c \rangle$,
where $res$ is the model checking result in $\{ \top, \bot \}$, and $c$ can be an empty set (when $M$ satisfies $\phi$), a \textit{definitive}-counterexample (when $M$ violates $\phi$), or a \textit{possible}-counterexample (when $M$ possibly-satisfies $\phi$).

\subsection{Unsatisfiable core}
\label{sub:UC}

Given a set of atomic propositions $AP$, $AP^{L}$ is the minimum set of elements such that for all $ p \in AP$, the following  holds: $p \in AP^{L}$ and $(\neg p) \in AP^{L}$.

\begin{definition}[\cite{SCHUPPAN2016155}]
Let $AP$ be a set of atomic propositions.
A Separated Normal Form (SNF) clause is an LTL formula, which is either an SNF initial clause, an SNF global clause or an SNF eventuality clause, where:
\begin{itemize}
	\item an \textit{SNF initial clause} $\underset{p \in P}{\bigvee}p$ where $P \subseteq AP^L$;
	\item an \textit{SNF global clause} $\LTLg(\underset{p \in P}{\bigvee}p \lor \LTLx(\underset{q \in Q}{\bigvee}q))$ where $P,Q \subseteq AP^L$;
	\item an \textit{SNF eventuality clause} $\LTLg(\underset{p \in P}{\bigvee}p \lor \LTLf(l))$ where $P \subseteq AP^L, l \in AP^L$.
\end{itemize}
\end{definition}
\noindent
Given a set $C$ of LTL formulae we denote the LTL formula $\underset{c \in C}\bigwedge c$ with $\eta(C)$.

\begin{definition}[\cite{SCHUPPAN2016155}]
Let $C$ be a set of SNF clauses. 
Then the formula $\eta(C)$ is in SNF.
\end{definition}

In the following we assume that the property $\phi$ is an LTL formula in SNF since any LTL formula can be transformed in an equally satisfiable SNF formula (for example, by using the procedure in~\cite{fisher2001clausal}). 

\begin{definition}[\cite{SCHUPPAN2016155}]
Let $C$ be a set of SNF clauses. 
Then $C$ is \emph{unsatisfiable} if the corresponding SNF formula $\eta(C)$ is unsatisfiable.
\end{definition}

\begin{definition}[\cite{SCHUPPAN2016155}]
Let $C$ be an unsatisfiable set of SNF clauses and 
let $C^{uc}$ be an unsatisfiable subset of $C$. 
Then
$C^{uc}$ is an \emph{Unsatisfiable Core (UC)} of $C$.
\end{definition}

\section{Revising and refining models}
\label{sec:topologicalproof}
First, we define how models can be revised and refined.
Then, we define the notion of \emph{topological proof}, that is used to describe why a property $\phi$ is satisfied in a KS $M$. 
Furthermore, we describe how the defined notion of proof can be exploited to show why a property is satisfied or possibly satisfied in a PKS.

\vspace{2mm}
\noindent
\textbf{Revisions and refinements.}
During a revision, a designer can add and remove states and transitions or change the labeling of the atomic propositions in some of the states of the structure.

\begin{definition}
\label{def:revision}
Let $M=\langle S,$ $R,$ $S_0,$ $AP,$ $L \rangle$ and 
$M^\prime=\langle S^\prime,$ $R^\prime,$ $S^\prime_0,$ $AP^\prime,$ $L^\prime \rangle$ be two PKSs.
Then $M^\prime$ is a \emph{revision} of $M$ if and only if $AP \subseteq AP^\prime$.
\end{definition}
Informally, the only constraint the designer has to respect during a revision is not to remove propositions from the set of atomic propositions.

\begin{restatable}{lemma}{refinementisrevision}
\label{refinementisrevision}
Let $M=\langle S,$ $R,$ $S_0,$ $AP,$ $L \rangle$ be a PKS and   
let $M^\prime=\langle S,$ $R,$ $S_0,$ $AP,$ $L^\prime\rangle$ be a refinement of $M$.
Then $M^\prime$ is a revision of $M$.
\end{restatable}

\vspace{2mm}
\noindent
\textbf{Topological proofs.}
The pursued proof is made of a set of clauses specifying certain topological properties of $M$, which ensure that the property is satisfied.

\begin{definition}
\label{def:ksclause}
Let  $M=\langle S,$ $R,$ $S_0,$ $AP,$ $L \rangle$ be a  PKS. 
A \emph{topological proof clause} (TP-clause) $\gamma$ for $M$ is either:
\begin{itemize}
\item  a \textit{topological proof propositional  clause} (TPP-clause), i.e., a triad $\langle s, \alpha, v \rangle$ where $s \in S$, $\alpha \in AP$, and $v \in \{ \LTLtrue, ?, \LTLfalse \}$; 
\item a \textit{topological proof transitions-from-state clause} (TPT-clause), i.e., an element $\langle s, T \rangle$, such that $s \in S, T \subseteq S$;
\item a \textit{topological proof initial-states clause} (TPI-clause), i.e., an element $\langle S_0 \rangle$.
\end{itemize}
\end{definition}

These clauses indicate \textit{topological properties} of a PKS $M$.
Informally, TPP-, TPT-, and TPI-clauses describe how states are labeled, how states are connected, and from which states the runs on the model begin, respectively.

\begin{definition}
\label{def:gammarelated}
Let $M=\langle S, R, S_0, AP, L \rangle$ be  a PKS and 
let $\Gamma$  be a set of TP-clauses for $M$. 
Then a \emph{$\Gamma$-related} PKS is a PKS $M^\prime=\langle S^\prime, R^\prime, S^\prime_0, AP^\prime, L^\prime \rangle$, such that the following conditions hold:
\begin{itemize}
\item $AP \subseteq AP^\prime$;
\item for every TPP-clause $\langle s, \alpha, v \rangle \in \Gamma$,  $v=L^\prime(s, \alpha)$;
\item for every TPT-clause $\langle s, T \rangle  \in \Gamma$, $T=\{s^\prime \in S^\prime | (s,s^\prime)\in\ R^\prime\}$;
\item for every TPI-clause $\langle S_0 \rangle \in \Gamma$, $S_0 = S^\prime_0$.
\end{itemize}
\end{definition}

Intuitively, a $\Gamma$-related PKS of $M$ is a PKS obtained from $M$ by changing any topological aspect that does not impact on the set of TP-clauses $\Gamma$.
Any transition whose source state is not the source state of a transition included in the TPT-clauses can be added or removed from the PKS and any value of a proposition that is not constrained by a TPP-clause can be changed.
States can be always added and they can be removed if they do not appear in any TPT-, TPP-, or TPI-clause. Initial states cannot be changed if $\Gamma$ contains a TPI-clause.

\begin{definition}
\label{def:topologicalproof}
Let $M=\langle S, R, S_0, AP, L \rangle$ be a PKS, 
let $\phi$ be an LTL property, 
let $\Omega$ be a set of TP-clauses, 
and let $x$ be a truth value in $\{\top,?\}$.
A set of TP-clauses $\Omega$ is an \emph{$x$-topological proof} (or $x$-TP) for $\phi$ in $M$ if  
$[M \models \phi] = x$ and
every $\Omega$-related PKS $M^\prime$ is such that $[M^\prime \models \phi]\geq x$.

\end{definition}

Note that, in this definition---and in the rest of the paper, when not differently specified---$\models$ indicates either $\models_3$ or $\models_T$.

Intuitively, an \emph{$x$-topological proof} is a set $\Omega$ such that every PKS $M^\prime$ that satisfies the conditions specified in Definition~\ref{def:gammarelated} is such that $[M^\prime \models \phi]\geq x$.
We call $\top$-TP a \emph{definitive-topological proof} and
$?$-TP a \emph{possible-topological proof}.
In Definition~\ref{def:topologicalproof} the operator $\geq$, assumes that values $\top, ?, \bot$ are ordered considering the classical information ordering $\top >\, ? > \bot$ among the truth values~\cite{bruns1999model}.

A $?$-TP for the PKS in Fig.~\ref{fig:motivatingmodel} and property $\phi_4$ is composed by the TP-clauses shown in Table~\ref{tab:motivatingExampleProof}.

\begin{table}[t]

\begin{tabular}{lc}
	\toprule
	\textbf{Proof generated for property $\phi_4$   } & \textbf{Clause}\\
	\midrule
	$\langle \mathit{OFF}, \mathit{suck}, \bot \rangle$,
	$\langle \mathit{IDLE}, \mathit{suck}, \bot \rangle$,
	$\langle \mathit{MOVING}, \mathit{suck},$ $?$ $\rangle$,
	$\langle \mathit{MOVING}, \mathit{move}, \top \rangle$  & TPP \\
	$\langle \mathit{OFF},\{\mathit{OFF}, \mathit{IDLE}\}  \rangle$,
	$\langle \mathit{IDLE},\{\mathit{OFF}, \mathit{IDLE}, \mathit{MOVING}\} \rangle$ & TPT \\
	$\langle \{\mathit{OFF}\} \rangle$ & TPI \\
	\bottomrule
\end{tabular}

 	\caption{An example of proof for the vacuum-cleaner example.}
\label{tab:motivatingExampleProof}
\vspace{-4mm}
\end{table}

\begin{definition}
	\label{def:omegaRevisionDef}
	Let $M$ and $M^\prime$ be two PKSs, 
	let $\phi$ be an LTL property, and 
	let $\Omega$ be an $x$-TP.
	Then $M^\prime$ is an \emph{$\Omega_x$-revision} of $M$ if 
	$M^\prime$ is $\Omega$-related to $M$.
\end{definition}
Intuitively, since the \emph{$\Omega_x$-revision} $M^\prime$ of $M$ is such that $M^\prime$ is $\Omega$-related w.r.t. $M$, it is obtained by changing the model $M$ while preserving the statements that are specified in the $x$-TP.
A revision $M^\prime$ of $M$ is \emph{compliant} with the $x$-TP for a property $\phi$ in $M$ if it is an \emph{$\Omega_x$-revision} of $M$. 

\begin{restatable}{lemma}{revisingpreservespropertysatisfaction}
\label{revisingpreservespropertysatisfaction}
Let $M$ be a PKS,
let $\phi$ be an LTL property such that $[M \models \phi]=\top$, and 
let $\Omega$ be a $\top$-TP. 
Then every $\Omega_\top$-revision $M^\prime$ is such that $[M^\prime \models \phi] =\top$.

\noindent
Let $M$ be a PKS,
let $\phi$ be an LTL property such that $[M \models \phi]=?$, and
let $\Omega$ be an $?$-TP. 
Then every $\Omega_?$-revision $M^\prime$ is such that $[M^\prime \models \phi] \in \{ \top, ?\}$.
\end {restatable}

\begin{proof}
We prove the first statement of the Lemma; the proof of the second statement is obtained by following the same steps.
If $\Omega$ is a $\top$-TP, 
it is a $\top$-TP for $\phi$ in $M^\prime$, 
since $M^\prime$ is an $\Omega_\top$-revision of $M$ (by Definition~\ref{def:omegaRevisionDef}).
Since $\Omega$ is a $\top$-TP for $\phi$ in $M^\prime$, then $[M^\prime \models \phi] \geq \top$ (by Definition~\ref{def:topologicalproof}).
\qed	
\end{proof}

\section{\NAME\ automated support}
\label{sec:automatedsupport}

This section describes the algorithms that support  the \textsc{analysis} and \textsc{re-check} phases of  \NAME .

\begin{algorithm}[t]
\captionof{algorithm}{The algorithm that supports the \textsc{analysis} phase.} 
\label{alg:analyze} 
\begin{algorithmic}[1] 
\Function{\analyze}{$M$, $\phi$}
	\State $\langle res, c \rangle$ = \checkalg$^\ast$($M_{opt} $, $\phi)$
	\label{line:check_opt}
	\If{$res == \bot$} \Return $\langle \bot, \{ c \} \rangle$
	\label{line:not_satisfied}
	\Else
	\State $\langle res^\prime, c^\prime \rangle$ = \checkalg$^\ast$($M_{pes} $, $\phi$)
	\label{line:check_pes}
		\If{$res^\prime ==\top$} \Return  $\langle \top, \{ \tpcompute(M, M_{pes} , \magicfunction(\phi)) \} \rangle$
		\label{line:satisfied}
		\Else
			\State \Return $\langle ?, \{ c^\prime, 	\tpcompute(M, M_{opt} $, $\magicfunction(\phi))\} \rangle$
				\label{line:possiblysatisfied}
		\EndIf
	\EndIf
\EndFunction
\end{algorithmic}
\end{algorithm}

\textbf{\textsc{Analysis.}} 
To analyze a PKS $M=\langle S, R,S_0,AP,L \rangle$, \NAME\ uses the three-valued model checking framework based on Theorem~\ref{th:threevaluedMC}.
The model checking result is provided as output by the \textsc{analysis} phase of \NAME , whose behavior is described in Algorithm~\ref{alg:analyze}.
Specifically, the algorithm returns a tuple $\langle x, y \rangle$, where $x$ is the verification result and $y$ is a set containing the counterexample, the topological proof or both of them.
The algorithm first checks whether the optimistic approximation $M_{opt}$ of the PKS $M$ satisfies property $\phi$ (Line~\ref{line:check_opt}).
If this is not the case, the property is violated by the PKS and the definitive-counterexample $c$ is returned (Line~\ref{line:not_satisfied}).
Then, it checks whether the pessimistic approximation $M_{pes}$ of the PKS $M$ satisfies property $\phi$ (Line~\ref{line:check_pes}).
If this is the case, the property is satisfied and the value $\top$ is returned along with the definitive-topological proof ($\top$-TP) computed by the \tpcompute\ procedure applied on the pessimistic approximation $M_{pes}$ and the property $\magicfunction(\phi)$.
If this is not the case, the property is possibly satisfied and the value $?$ is returned along with the possible-counterexample $c'$ and the possible-topological proof ($?$-TP) computed by the \tpcompute\ procedure applied to $M_{opt}$ and $\magicfunction(\phi)$.

\begin{algorithm}[t]
\captionof{algorithm}{Compute Topological Proofs} 
\label{alg:computetpp} 
\begin{algorithmic}[1] 
\Function{\tpcompute}{$M$, $\mathcal{A}$, $\psi$}
	\State $\eta(C)= \tosnf (\mathcal{A}, \psi)$ \label{step1}
	\State $C^{uc} =\getuc (\eta(C))$ \label{step2}
	\State $TP= \gettp (M, C^{uc})$ \label{step3}
	\State \Return $TP$
\EndFunction
\end{algorithmic}
\end{algorithm}

The procedure \tpcompute\ (Compute Topological Proofs) to compute $x$-TPs is described in Algorithm~\ref{alg:computetpp}. 
It takes as input a PKS $M$, its optimistic/pessimistic approximation, i.e., the KS $\mathcal{A}$, and an LTL formula $\psi$---satisfied in $\mathcal{A}$--- corresponding to the transformed property \magicfunction($\phi$). The three steps are described in the following.

\tosnf . \textit{Encoding of the KS A and the LTL $\psi$ formula into an LTL formula in SNF $\eta(C)$.}
The KS $\mathcal{A}$ and the LTL  formula $\psi$ are used to generate an SNF formula 
$\eta(C_{\mathcal{A}} \cup C_{\psi})$, where $C_{\mathcal{A}}$ and $C_{\psi}$ are sets of SNF clauses obtained from the KS $\mathcal{A}$ and the LTL formula $\psi$. 
The clauses in  $C_{\psi}$ are computed from $\psi$ as specified in~\cite{SCHUPPAN2016155}.
The set of clauses that encodes the KS is $C_{\mathcal{A}}=C_{\mathit{KS}} \cup C_{\mathit{REG}}$, where $C_{\mathit{KS}}=\{ c_i \}  \cup CR_{\mathcal{A}} \cup CL_{\top,\mathcal{A}} \cup CL_{\bot,\mathcal{A}}$
and $c_i$, $CR_{\mathcal{A}}$, $CL_{\top,\mathcal{A}}$ and $CL_{\bot,\mathcal{A}}$
 are defined as specified in Table~\ref{tab:kstosnf}.
Note that the clauses in $C_{\mathcal{A}}$ are defined on the set of atomic propositions $AP_S=AP_{\mathcal{A}} \cup \{ p(s) | s \in S_{\mathcal{A}}\}$, i.e., $AP_S$ includes an additional atomic proposition $p(s)$ for each state $s$, which is true when the KS is in state $s$.

\begin{table}[t]
\caption{Rules to transform the KS in SNF formulae.}
\label{tab:kstosnf}
\begin{tabular}{l}
\toprule
$c_i=\underset{s \in S_{0,\mathcal{A}}}{\bigvee} p(s)$ \\ The KS is initially  in one of its initial states.\\
\midrule
$CR_{\mathcal{A}}=\{\LTLg(\neg p(s) \lor\LTLx( \underset{(s, s^\prime) \in R_{\mathcal{A}}}{\bigvee} p(s^\prime)) ) \mid s \in S_{\mathcal{A}} \}$ \\ If the KS is in state $s$ in the current time instant, in the next time instant it is in \\ one of its successors $s^\prime$ of $s$.\\
\midrule
$CL_{\top,\mathcal{A}}=\{\LTLg(\neg p(s) \lor \alpha ) \mid s \in S_{\mathcal{A}}, \alpha \in AP_{\mathcal{A}}, L_{\mathcal{A}}(s, \alpha)=\top \}$ \\
If the KS is in state $s$ s.t. $L_{\mathcal{A}}(s, \alpha)=\top$, the atomic proposition $\alpha$ is true.\\
\midrule
$CL_{\bot,\mathcal{A}}=\{\LTLg(\neg p(s) \lor \neg \alpha ) \mid s \in S_{\mathcal{A}}, \alpha \in AP_{\mathcal{A}}, L_{\mathcal{A}}(s, \alpha)=\bot \}.$ \\ If the KS is in state $s$ s.t. $L_{\mathcal{A}}(s, \alpha)=\bot$, the atomic proposition $\alpha$ is false.\\
\midrule
$C_{\mathit{REG}}=\{\LTLg(\neg p(s) \lor \neg p(s^\prime)) \mid s,s^\prime \in S_{\mathcal{A}} \text{ and } s\neq s^\prime\}$ \\ Ensures that the KS is in at most one  state at any time.\\
\toprule
\end{tabular}
\vspace{-4mm}
\end{table}

\getuc . \textit{Computation of the unsatisfiable core (UC) $C^{uc}$ of $C$.}
Since the property $\psi$ is satisfied on $\mathcal{A}$, as recalled in Section~\ref{sec:background}, $\eta(C_{\mathcal{A}} \cup C_{\psi})$ is unsatisfiable and the computation of its UC core is performed as specified in~\cite{SCHUPPAN2016155}.
The procedure returns an SNF formula $\eta(C_{\mathit{KS}}^\prime \cup C_{\mathit{REG}}^\prime \cup C^\prime_{\psi})$ that is unsatifiable and such that $C_{\mathit{KS}}^\prime \subseteq C_{\mathit{KS}}$,
$C_{\mathit{REG}}^\prime \subseteq C_{\mathit{REG}}$ and $C_{\psi}^\prime \subseteq C_{\psi}$.

\gettp . \textit{Analysis of the UC $C^{uc}$ and extraction of the topological proof.}
Formula $\eta(C_{\mathcal{A}}^\prime \cup C^\prime_{\psi})$, where $C_{\mathcal{A}}^\prime=C_{\mathit{KS}}^\prime \cup C_{\mathit{REG}}^\prime$,
 contains clauses regarding the KS ($C^\prime_{\mathit{KS}}$), the fact that the model is a KS ($C^\prime_{\mathit{REG}}$), and the property of interest ($C_{\psi}^\prime$) that made the formula $\eta(C_{\mathcal{A}} \cup C_{\psi})$ unsatisfiable. 
Since we are interested in clauses related to the KS that caused unsatisfiability, we extract the topological proof $\Omega$, whose topological proof clauses are obtained from the clauses in $C_{\mathit{KS}}^\prime$ as specified in Table~\ref{tab:snftotp}.
Since the set of atomic propositions of $\mathcal{A}$ is $AP_{\mathcal{A}}= AP \cup \overline{AP}$, in the table we use $\alpha$ for propositions in $AP$ and $\overline{\alpha}$ for propositions in $\overline{AP}$.

\begin{table}[t]
\caption{Rules to extract the TP-clauses from the UC SNF formula.}
\label{tab:snftotp}
\begin{tabular}{ p{0.4\textwidth} p{0.4\textwidth} p{0.2\textwidth}}
\toprule
\textbf{SNF clause} & \textbf{TP clause} & \textbf{TP clause type} \\
\toprule
$c_i=\underset{s \in S_{0,\mathcal{A}}}{\bigvee} p(s)$ & $\langle S_0 \rangle$ & TPI-clause \\
\midrule
$\LTLg(\neg p(s) \lor \LTLx(\underset{(s, s^\prime) \in R_{\mathcal{A}}}{\bigvee} p(s^\prime) )$ & $\langle s, T \rangle$ where $T=\{s^\prime | (s,s^\prime) \in R \}$  & TPT-clause \\
\midrule
$\LTLg(\neg p(s) \lor \alpha )$ & $\langle  s, \alpha, L(s, \alpha) \rangle$  & TPP-clause \\
\midrule
$\LTLg(\neg p(s) \lor \neg \alpha )$ & $\langle  s, \alpha, comp(L(s, \alpha)) \rangle$  & TPP-clause \\
\midrule
$\LTLg(\neg p(s) \lor \overline\alpha )$ & $\langle  s, \alpha, comp(L(s, \alpha)) \rangle$ & TPP-clause \\
\midrule
$\LTLg(\neg p(s) \lor \neg \overline\alpha )$ & $\langle  s, \alpha, L(s, \alpha) \rangle$ & TPP-clause \\
\toprule
\end{tabular}
\vspace{-4mm}
\end{table}

Note that elements in $C_{\mathit{REG}}^\prime$ are not considered  in the TP computation.
Indeed, given an SNF clause $\LTLg(\neg p(s) \lor \neg p(s^{\prime}))$, either state $s$ or $s^\prime$ is included in other SNF clauses, thus it will be mapped on TP-clauses that will be preserved in the model revisions.

	\begin{lemma}
				\label{lemma:creg}
Let $\mathcal{A}$ be a KS
and let $\psi$ be an LTL property.
Let also $\eta(C_{\mathcal{A}} \cup C_{\psi})$ be the SNF formula computed in the step \tosnf of the algorithm, where $C_{\mathcal{A}}=C_{\mathit{REG}}\cup C_{\mathit{KS}}$, and 
let $C_{\mathcal{A}}^\prime \cup  C_{\psi}^\prime$ be an unsatisfiable core,  where $C_{\mathcal{A}}^\prime=C_{\mathit{REG}}^\prime\cup C_{\mathit{KS}}^\prime$.
				Then, if $\LTLg(\neg p(s) \lor \neg p(s^{\prime})) \in C^\prime_{\mathit{REG}}$,
				either
				\begin{enumerate}
				\item[(i)] there exists an SNF clause in $C_{\mathit{KS}}^\prime$  that predicates on state $s$ (or on state $s^{\prime}$);
				\item[(ii)] $C_{\mathcal{A}}^{\prime\prime}  \cup  C_{\psi}^\prime$, s.t. $C_{\mathcal{A}}^{\prime\prime}=C_{\mathcal{A}}^{\prime} \setminus \{ \LTLg(\neg p(s) \lor \neg p(s^{\prime}))  \}$, is  an unsatisfiable core of $\eta(C^\prime_{\mathcal{A}} \cup C^{\prime}_{\psi})$.
				\end{enumerate}
		\end{lemma}
		
			\begin{proof}				
We indicate $\LTLg(\neg p(s) \lor \neg p(s^{\prime}))$ as $\tau(s,s^\prime)$.
Assume per absurdum that conditions (i) and (ii) are violated, i.e., no SNF clause in $C^\prime_{\mathit{KS}}$  predicates on state $s$ or $s^\prime$ and 
$C_{\mathcal{A}}^{\prime\prime}  \cup  C_{\psi}^\prime$ is not an unsatisfiable core of $C^\prime_{\mathcal{A}} \cup C^{\prime}_{\psi}$.
Since $C_{\mathcal{A}}^{\prime\prime}  \cup  C_{\psi}^\prime$ is not an unsatisfiable core of $C^\prime_{\mathcal{A}} \cup C^{\prime}_{\psi}$, $C_{\mathcal{A}}^{\prime\prime} \cup C^{\prime}_{\psi}$ is satisfiable since $C_{\mathcal{A}}^{\prime\prime} \subset C^\prime_{\mathcal{A}}$.
Since $C_{\mathcal{A}}^{\prime\prime}  \cup  C_{\psi}^\prime$ is satisfiable,
$C_{\mathcal{A}}^{\prime}  \cup  C_{\psi}^\prime$ s.t. $C_{\mathcal{A}}^{\prime}=C_{\mathcal{A}}^{\prime\prime} \cup  \{\tau(s,s^\prime)\}$ must also be satisfiable. Indeed, it does not exist any SNF clause that predicates on state $s$ (or on state $s^{\prime}$) and, in order to generate a contradiction, the added SNF clause must generate it with the SNF clauses obtained from the LTL property $\psi$.
This is a contradiction.
Thus, conditions (i) and (ii) must be satisfied.
			\qed
			\end{proof}

\vspace{2mm}
The \analyze\ procedure in Algorithm~\ref{alg:analyze} has shown how we obtain a TP for a PKS by 
first computing the related optimistic or pessimistic approximation (i.e., a KS) and 
then exploiting the computation of the TP for this KS.

\begin{restatable}{theorem}{topologicalproofcorrectness}
\label{topologicalproofcorrecntess}
Let $M=\langle S, R,S_0,AP,L \rangle$ be a PKS, 
let $\phi$ be an LTL property, 
and let $x\in\{\top,?\}$ be an element such that $[M \models_{3} \phi]=x$.
If the procedure \analyze\, applied to the PKS $M$ and the LTL property $\phi$, returns a TP $\Omega$, this is an $x$-TP for $\phi$ in $M$.
\end{restatable}

\begin{proof}
	Assume  that the
	\analyze\ procedure returns the value $\top$ and a $\top$-TP.
	We show that every $\Omega$-related PKS $M^\prime$ is such that $[M^\prime \models \phi]\geq x$ (Definition~\ref{def:topologicalproof}).
	If \analyze\  returns the value $\top$, it must be that
	$M_{\mathit{pes}}\models^\ast \phi$ by Lines~\ref{line:check_pes} and~\ref{line:satisfied} of Algorithm~\ref{alg:analyze}.
	Furthermore, by Line~\ref{line:satisfied}, $\psi=\magicfunction(\phi)$ and $\mathcal{A}=M_{\mathit{pes}}$.
	Let   $N=\langle S_N, R_N, S_{0,N}, AP_N,L_N \rangle$ be a PKS $\Omega$-related to $M$.
	Let 	$\eta(C_{\mathcal{A}} \cup C_{\psi})$  be the SNF formula associated with $\mathcal{A}$ and $\psi$ and 
	let $\eta(C_{\mathcal{B}} \cup C_{\psi})$  be the SNF formula associated with $\mathcal{B}=N_{\mathit{pes}}$ and $\psi$.
	Let us consider an UC $C^\prime_{\mathcal{A}}$ $\cup C_{\psi}^\prime$ of $C_{\mathcal{A}} \cup C_{\psi}$,
	where $C^\prime_{\mathcal{A}}=C_{\mathit{KS}}^\prime \cup C_{\mathit{REG}}^\prime$, 
	$C_{\mathit{KS}}^\prime \subseteq C_{\mathit{KS}}$,
	$C_{\mathit{REG}}^\prime \subseteq C_{\mathit{REG}}$, and 
	$C_{\psi}^\prime \subseteq C_{\psi}$.
	We show that $C^\prime_{\mathcal{A}} \subseteq C_{\mathcal{B}}$ and $C_{\psi}^\prime \subseteq C_{\psi}$, i.e., the UC is also an UC for the 
SNF formula associated with the approximation $\mathcal{B}$	of the  PKS $N$.
	\begin{itemize}		
	   \item $C_{\psi}^\prime \subseteq C_{\psi}$ is trivially verified since property $\psi$ does not change.
		\item $C^\prime_{\mathcal{A}} \subseteq C_{\mathcal{B}}$, i.e., 
		$(C^\prime_{\mathit{KS}} \cup C_{\mathit{REG}}^\prime) \subseteq C_{\mathcal{B}}$.
		By Lemma~\ref{lemma:creg} we can avoid considering $C_{\mathit{REG}}^\prime$.
By construction (see Line~\ref{step1} of Algorithm~\ref{alg:computetpp}) any clause $c \in C_{\mathit{KS}}^\prime$ belongs to one rule among $CR$, $CL_{pes,\top}$, $ CL_{pes,\bot}$ or $c=c_i$:
		\begin{itemize}
			\item if $c = c_i$ then, by the rules in Table~\ref{tab:snftotp},  there is a TPI-clause $\{S_0\} \in \Omega$. By Definition~\ref{def:gammarelated}, $S_0=S_0^\prime$.
			Thus, $c_i \in C_{\mathcal{B}}$ since $N$ is $\Omega$-related to $M$.
			\item if $c \in CR$ then, by rules in Table~\ref{tab:snftotp}, there is a TPT-clause $\langle s, T \rangle \in \Omega$ where $s\in S$ and $T \subseteq R$.					
			By Definition~\ref{def:gammarelated}, $T=\{s^\prime \in S^\prime | (s,s^\prime)\in R^\prime\}$.
			Thus, $c \in C_{\mathcal{B}}$ since $N$ is $\Omega$-related to $M$. 
			\item if $c \in CL_{\mathcal{A},\top}$ or $c \in CL_{\mathcal{A},\bot}$, by rules in Table~\ref{tab:snftotp},  there is a TPP-clause $\langle s, \alpha, L(s,\alpha) \rangle \in \Omega$ where $s\in S$ and $\alpha \in AP$.			
			By Definition~\ref{def:gammarelated}, $L^\prime(s,\alpha)=L(s,\alpha)$.
			Thus, $c \in C_{\mathcal{B}}$ since $N$ is $\Omega$-related to $M$.			
		\end{itemize}	
	\end{itemize}
	Since $N$ is $\Omega$-related to $M$, it has preserved the elements of $\Omega$. Thus $C^\prime_{\mathcal{A}}\cup C_{\psi}^\prime$ is also an UC of $C_{\mathcal{B}}$.
	It follows that $[N \models \phi] = \top$. 
	
\vspace{2mm}	
The proof from the case in which		\analyze\ procedure returns the value $?$ and a $?$-TP can be derived from the first case.
	\qed	
\end{proof}

\textbf{\textsc{Re-check.}} 
Let  $M$ be a PKS. The \textsc{re-check} algorithm verifies whether a revision $M^\prime$ of $M$ is an $\Omega$-revision.
Let  $M=\langle S, R, S_0, AP, L \rangle$ be a  PKS, 
let  $\Omega$ be an $x$-TP  for $\phi$ in $M$, and 
let $M^\prime=\langle S^\prime, R^\prime, S_0^\prime, AP^\prime, L^\prime \rangle$ be a revision of $M$.
The \textsc{re-check} algorithm returns \texttt{true} if and only if the following holds:
\begin{itemize}
	\item $AP \subseteq AP^\prime$;
\item for every TPP-clause  $\langle s, \alpha, v \rangle \in \Omega$, $v=L^\prime(s, \alpha)$;
\item for every TPT-clause $\langle s, T \rangle \in \Omega$, 
$T=\{s^\prime \in S^\prime|(s,s^\prime)\in\ R^\prime\}$;
\item for every TPI-clause $\langle S_0 \rangle \in \Omega$,   
$S_0 = S^\prime_0$.
\end{itemize}

\begin{lemma}
	\label{lemma:omega-related}
	Let $M=\langle S, R, S_0, AP, L \rangle$ and $M^\prime=\langle S^\prime, R^\prime,$ $S^\prime_0,$ $AP^\prime, L^\prime \rangle$ be two PKSs
	and let $\Omega$ be an $x$-TP.
	The \textsc{re-check} algorithm returns \texttt{true} if and only if $M^\prime$ is $\Omega$-related to $M$.
\end{lemma}
\begin{proof}
Since $M^\prime$ is $\Omega$-related to $M$, 
the conditions of Definition~\ref{def:gammarelated} hold.
Each of these conditions corresponds to a condition of the \textsc{re-check} algorithm. 
Thus, if $M^\prime$ is $\Omega$-related to $M$, the \textsc{re-check} returns \texttt{true}.
Conversely, if \textsc{re-check} returns \texttt{true}, each condition of the algorithm is satisfied and, since each of this conditions is mapped to a condition of Definition~\ref{def:gammarelated},
$M^\prime$ is $\Omega$-related to $M$. \qed	
\end{proof}

\begin{restatable}{theorem}{topologicalproofcorrecntess}
\label{recheckcorrectness}
Let $M$ be a PKS, 
let $\phi$ be a property, 
let $\Omega$ be an $x$-TP for $\phi$ in $M$ where $x \in \{\top,?\}$, and
let $M^\prime$ be a revision of $M$.
The \textsc{re-check} algorithm returns \texttt{true} if and only if $M^\prime$ is an $\Omega$-revision of $M$.
\end{restatable}

\begin{proof}	
	By applying Lemma~\ref{lemma:omega-related}, the \textsc{re-check} algorithm returns \texttt{true} if and only if $M^\prime$ is 
	$\Omega$-related to $M$.
	By Definition~\ref{def:omegaRevisionDef},
	since $\Omega$ is an $x$-TP, the \textsc{re-check} algorithm returns \texttt{true} if and only if $M^\prime$ is an $\Omega$-revision of $M$. \qed
\end{proof}

The \textsc{analysis} and \textsc{re-check} algorithms assume that the three-valued LTL semantics is considered.
Indeed, the algorithm upon which the proof generation framework is developed is the  three-valued model checking algorithm~\cite{bruns2000generalized}.
However, the proposed results are also valid considering the thorough semantics if the properties of interest are self-minimizing.
This is not a strong limitation since, as shown in~\cite{godefroid2005model}, most practically useful LTL formulae are self-minimizing.
Future work will consider how to extend the \textsc{analysis} and \textsc{re-check} to completely support the thorough  LTL semantics.
 \section{Evaluation}
\label{sec:evaluation}

To evaluate how \NAME\ supports designers, we implemented \NAME\ as a Scala stand alone application which is available at \url{http://goo.gl/V4fSjG}.
Specifically, we considered the following research questions:
\begin{enumerate*}
\item[] \textbf{RQ1:} How does the \textsc{analysis} help in creating models revisions?
\item[] \textbf{RQ2:} How does the \textsc{re-check} help in creating models revisions?
\end{enumerate*}

\vspace{2mm}
\noindent
\textbf{Evaluation setup.}
To answer RQ1 and RQ2 we considered a set of example PKSs proposed in literature to evaluate the $\chi\mathit{Chek}$~\cite{1201295} model checker.
These examples are divided into three categories.
Each category represents an  aspect of a phone call, i.e., 
the \texttt{callee}, \texttt{caller}, and \texttt{caller-callee} categories include PKSs modeling respectively the callee process, the caller process, and the overall calleer-callee process.
To simulate iterative design performed using the \NAME\ framework all examples were slightly modified. 
Uncertainty on the transitions was removed in all the PKSs of the considered examples.
In addition, the examples in the \texttt{callee} category were modified such that the designer iteratively generates refinements of PKSs, i.e., $?$ is assigned to atomic propositions to generate abstractions of the final KS \texttt{callee-4}.
Instead, the examples in the \texttt{caller-callee} category were modified since the original examples were proposed to evaluate designer disagreement about the value to assign to a proposition in a state~\cite{Chechik2006}. 
Thus, in a state $s$ a proposition could be assigned to multiple values in the set $\{\top, ?, \bot\}$.
To generate a PKS useful for our evaluation, we used the following transformations: 
when, in a state $s$, a proposition was assigned to the values $\{\top,\bot \}$, $\{ \top, ?\}$, or $\{\bot, ? \}$   these values were replaced with $?$; when, in a state $s$, a proposition was assigned values $\{\top,\top\}$ (resp. $\{\bot,\bot\}$),   these values were replaced with $\top$ (resp. $\bot$).

We defined a set of properties based on well known LTL property patterns~\cite{dwyer1998property}.
This choice was motivated by the fact that the original properties used in the examples were specified in Computation Tree Logic (CTL), which is not supported by \NAME. 

The defined properties were inspired by the original properties from the examples and are listed in Table~\ref{tab:properties}.

\color{red}
\begin{table}[t]
	\small
	\caption{Properties considered in the evaluation} 	
	\begin{tabular}{llll}
		\toprule
		$\phi_1$: &&& $\LTLg (\neg \mathit{OFFHOOK}) \lor (\neg\mathit{OFFHOOK}~ \LTLu ~\mathit{CONNECTED})$ \\
		$\phi_2$: &&& $\neg\mathit{OFFHOOK}~\LTLw~(\neg\mathit{OFFHOOK} \land \mathit{CONNECTED})$\\
		$\phi_3$: &&& $\LTLg (\mathit{CONNECTED} \rightarrow \mathit{ACTIVE})$ \\
		$\phi_4$: &&& $\LTLg (\mathit{OFFHOOK} \land \mathit{ACTIVE} \land \neg\mathit{CONNECTED} \rightarrow \LTLx(\mathit{ACTIVE}))$ \\
		$\phi_5$ &&& $\LTLg(\mathit{CONNECTED} \rightarrow \LTLx(\mathit{ACTIVE}))$ \\
		\midrule
		$\psi_1$: &&& $\LTLg(\mathit{CONNECTED} \rightarrow \mathit{ACTIVE})$ \\
		$\psi_2$: &&& $\LTLg(\mathit{CONNECTED} \rightarrow \LTLx(\mathit{ACTIVE}))$\\
		$\psi_3$: &&& $\LTLg(\mathit{CONNECTED}) \lor (\mathit{CONNECTED}~\LTLu~\neg\mathit{OFFHOOK})$ \\
		$\psi_4$: &&& $\neg\mathit{CONNECTED}~\LTLw~(\neg\mathit{CONNECTED} \land \mathit{OFFHOOK})$ \\
		$\psi_5$: &&& $\LTLg(\mathit{CALLEE\_SEL} \rightarrow \mathit{OFFHOOK})$ \\
		\midrule
		$\eta_1$: &&& $\LTLg((\mathit{OFFHOOK} \land \mathit{CONNECTED}) \rightarrow \LTLx(\mathit{OFFHOOK} \lor \neg\mathit{CONNECTED}))$ \\
		$\eta_2$: &&& $\LTLg(\mathit{CONNECTED}) \lor (\mathit{CONNECTED}~\LTLw~\neg\mathit{OFFHOOK})$\\
		$\eta_3$: &&& $\neg\mathit{CONNECTED}~\LTLw~(\neg\mathit{CONNECTED} \land \mathit{OFFHOOK})$ \\
		$\eta_4$: &&& $\LTLg(\mathit{CALLEE\_FREE} \lor \mathit{LINE\_SEL})$ \\
		$\eta_5$: &&& $\LTLg(\LTLx(\mathit{OFFHOOK}) \land \neg\mathit{CONNECTED})$ \\
		\bottomrule
	\end{tabular}
	\label{tab:properties}
\end{table}
\color{black}

\vspace{2mm}
\noindent
\textbf{RQ1.} To answer RQ1 we checked how the proofs output by the \textsc{analysis} algorithm were useful in producing PKSs revisions.
To evaluate the usefulness we checked how easy it was to analyze the property satisfaction on the proofs w.r.t. the original models.
This was done by comparing the size of the proofs and the size of the original models.
The size of a PKS $M=\langle S, R, S_0, AP, L \rangle$ was defined as $|M|=|AP|*|S|+|R|+|S_0|$.
The size of a proof $\Omega$ was defined as $|\Omega|=\underset{c \in \Omega}{\sum} |c|$ where: 
$|c|=1$ if $c=\langle s,\alpha,v \rangle$;
$|c|=|T|$ if $c=\langle s,T \rangle$, and 
$|c|=|S_0|$ if $c=\langle S_0\rangle$.
Table~\ref{tab:experiments} summarizes the obtained results, indicated in the columns under the label `RQ1'.  
We  show the cardinalities
$|S|$, $|R|$ and $|AP|$ of the sets of states, transitions, and atomic propositions of each considered PKS $M$, the number $|?|$ of couples of a state $s$ with an atomic proposition $\alpha$ such that $L(s,\alpha)=?$, the total size $|M|$ of the model, and the size $|\Omega_p|_x$ of the proofs, where $p$ indicates the considered LTL property and $x$ indicates whether $p$ is satisfied ($x=\top$) or possibly satisfied ($x=?$).
Cells labeled with the symbol $\times$ indicate that a property was not satisfied  in that model and thus a proof was not produced by the \textsc{analysis} algorithm.
It can be observed that  the size of the proof was always lower than the size of the initial model.
This is an indicator that shows that proofs are easier to understand than the original models, since they include a subset of the elements of the models that ensure that a property is satisfied (resp. possibly satisfied).

\begin{table}[t]
	\caption{Cardinalities $|S|$, $|R|$, $|AP|$, $|?|$, and $|M|$ are those of the evaluated model $M$.
	$|\Omega_p|_x$ is the size of proof $\Omega_p$ for a property $p$;
$x$ indicates if $\Omega_p$ is a $\top$-TP or a $?$-TP.} 
\begin{tabular}{lcccccccccc|ccccl}
\toprule
 & \multicolumn{10}{c}{\textbf{RQ1}} & \multicolumn{5}{c}{\textbf{RQ2}}\\
\toprule
Model & $|S|$ & $|R|$ & $|AP|$ & $|?|$ & $|M|$ & $|\Omega_{\phi_1}|$ & $|\Omega_{\phi_2}|$ & $|\Omega_{\phi_3}|$ & $|\Omega_{\phi_4}|$ & $|\Omega_{\phi_5}|$ & $\phi_1$ & $\phi_2$ & $\phi_3$ &$\phi_4$ & $\phi_5$ \\
\midrule
callee-1 & 5 & 15 & 3 & 7 & 31 & $7_?$ & $9_?$ & $21_?$ & $23_?$ & $23_?$ & \textbf{-} & \textbf{-}  & \textbf{-} & \textbf{-} & \textbf{-} \\
callee-2 & 5 & 15 & 3 & 4 & 31 & $7_?$ & $9_?$ & $21_?$ & $22_\top$ & $\times$ & \cmark\ & \cmark\ & \cmark\ & \cmark\ & \xmark\ \\
callee-3 & 5 & 15 & 3 & 2 & 31 & $7_?$ & $9_?$ & $21_?$ & $23_\top$ & $\times$ & \cmark\ & \cmark\ & \cmark\ & \cmark\ & \textbf{-}\\
callee-4 & 5 & 15 & 3 & 0 & 31 & $\times$ & $\times$ & $23_\top$ & $21_\top$ & $\times$ & \xmark\ & \xmark\ & \cmark\ & \cmark\ & \textbf{-} \\
\midrule
Model & $|S|$ & $|R|$ & $|AP|$ & $|?|$ & $|M|$ & $|\Omega_{\psi_1}|$ & $|\Omega_{\psi_2}|$ & $|\Omega_{\psi_3}|$ & $|\Omega_{\psi_4}|$ & $|\Omega_{\psi_5}|$ & $\psi_1$ & $\psi_2$ & $\psi_3$ &$\psi_4$ & $\psi_5$ \\
\midrule
caller-1 & 6 & 21 & 5 & 4 & 52 & $28_?$ & $\times$ & $2_\top$ & $9_?$ & $28_?$ & \textbf{-} & \textbf{-} & \textbf{-} & \textbf{-} & \textbf{-} \\
caller-2 & 7 & 22 & 5 & 4 & 58 & $30_?$ & $\times$ & $2_\top$ & $9_?$ & $30_?$ & \cmark\ & \textbf{-} & \cmark\ & \cmark\ & \cmark\ \\
caller-3 & 6 & 19 & 5 & 1 & 50 & $26_\top$ & $28_\top$ & $2_\top$ & $11_\top$ & $26_\top$ & \cmark\ & \textbf{-} & \cmark\ & \cmark\ & \cmark\ \\
caller-4 & 6 & 21 & 5 & 0 & 52 & $28_\top$ & $\times$ & $2_\top$ & $9_\top$ & $28_\top$ & \cmark\ & \xmark\ & \cmark\ & \cmark\ & \cmark\ \\
\midrule
Model & $|S|$ & $|R|$ & $|AP|$ & $|?|$ & $|M|$ & $|\Omega_{\eta_1}|$ & $|\Omega_{\eta_2}|$ & $|\Omega_{\eta_3}|$ & $|\Omega_{\eta_4}|$ & $|\Omega_{\eta_5}|$ & $\eta_1$ & $\eta_2$ & $\eta_3$ &$\eta_4$ & $\eta_5$ \\
\midrule
caller-callee-1 &  6 & 30 & 6 & 30 & 61 & $37_?$ & $2_\top$ & $15_?$ & $37_?$ & $\times$ & \textbf{-} & \textbf{-} & \textbf{-} & \textbf{-} & \textbf{-} \\
caller-callee-2 &  7 & 35 & 6 & 36 & 78 & $43_?$ & $2_\top$ & $18_?$ & $43_?$ & $\times$ & \cmark\ & \cmark\ & \cmark\ & \cmark\ & \textbf{-} \\
caller-callee-3 &  7 & 45 & 6 & 38 & 88 & $53_?$ & $2_\top$ & $53_?$ & $53_?$ & $53_?$ & \cmark\ & \cmark\ & \cmark\ & \cmark\ & \textbf{-} \\
caller-callee-4 &  6 & 12 & 4 & 0 & 42 & $\times$ & $\times$ & $\times$ & $19_\top$ & $\times$ & \xmark\ & \xmark\ & \xmark\ & \cmark\ & \xmark\ \\
\bottomrule
\end{tabular}
	\label{tab:experiments}
\end{table}

\vspace{2mm}
\noindent
\textbf{RQ2.} 
To answer RQ2 we checked how the results produced by the \textsc{re-check} algorithm were useful in producing PKSs revisions.
To evaluate the usefulness we assumed that, for each category of examples, the designer produced revisions following the order specified in Table~\ref{tab:experiments}.
The columns under the label `RQ2' contain the different properties that have been analyzed for each category.
A cell contains \cmark\ if the \textsc{re-check} was passed by the considered revised model, i.e., a \texttt{true} value was returned by the \textsc{re-check} algorithm, \xmark\ otherwise. The \textit{dash} symbol \textbf{-} is used when the model of the correspondent line is not a revision (i.e., the first model of each category) or when the observed property was false in the previous model, i.e., an $x$-TP was not produced.
We inspect results produced by the \textsc{re-check} algorithm to evaluate their utility in verifying if revisions were violating the proofs.
A careful observation of Table \ref{tab:experiments} reveals that, in many cases, the \NAME\ \textsc{re-check} notifies the designer that the proposed revision violates some of the clauses contained in the $\Omega$-proof. 
This suggests that the \textsc{re-check} is effective in helping the designer in creating model revisions.

\vspace{2mm}
\noindent
\textbf{Threats to validity.} 
The changes applied to the existing models are a threat to construct validity since they may generate models which are not realistic. 
To mitigate this threat, we referred to textual explanations in the papers in which the examples were presented to generate reasonable revisions. 
Biases in the creation of PKSs are a threat to internal validity.
To mitigate this threat, we designed our models starting from already existing models. 
The limited number of examples is a threat to external validity. 
To mitigate this threat, we verified that all the possible output cases of \NAME\ were obtained at least once.

\vspace{2mm}
\noindent
\textbf{Scalability.}
Three-valued model checking is as expensive as classical model checking~\cite{bruns1999model}, which is commonly employed in real world problems analysis~\cite{Woodcock09}. 
Unsatisfiability checking and UCs computation has been employed to verify digital hardware and software systems~\cite{hustadt2003trp++}.
The \textsc{analysis} phase of \NAME\ simply combines three-valued model checking and UCs computation, therefore
its scalability improves as the performance of the employed integrated frameworks enhances.
Future investigation may evaluate if the execution of the \textsc{re-check} algorithm (which is a simple syntactic check) speeds up the verification framework by avoiding the re-execution of the \textsc{analysis} algorithm in the cases in which revisions satisfy the proofs. 
 \section{Related work}
\label{sec:related}
\begin{sloppypar}
Partial knowledge has been considered in requirement analysis and elicitation~\cite{menghi2017integrating,menghi2017cover,letier2008deriving},
in novel robotic planners~\cite{10.1007/978-3-319-95582-7_24,menghi2018towards},
and in the production of software models that satisfy a set of desired properties~\cite{uchitel2009synthesis,uchitel2013supporting,famelis2012partial,albarghouthi2012under}.
Several researchers analyzed the model checking problem for partially specified systems~\cite{menghi2016dealing,chechik2004multi},
some considering three-valued~\cite{larsen1988modal,godefroid2001abstraction,bruns1999model,bruns2000generalized,godefroid2011ltl}, 
others multi-valued~\cite{gurfinkel2003multi,bruns2004MCmultivalued} scenarios.
Other works apply model checking to incremental program development~\cite{henzinger2003extreme,beyer2007software}.
However, all these model checking approaches do not provide an \emph{explanation} on why a property is satisfied, by means of a \emph{certificate} or \emph{proof}.
Although several works have tackled this problem~\cite{Bernasconi2017,cleaveland2002evidence,PZ01,PPZ01,griggio2018certifying,deng2017witnessing},
they aim mostly to automate proof reproducibility rather than actually helping the design process and they usually produce deductive proofs which are different from the topological proofs presented in this paper.
\end{sloppypar}

Tao and Li~\cite{tao2017complexity} propose a theoretical solution to a related issue: \textit{model repair} is the problem of finding the minimum set of states in a KS which makes a formula satisfiable. However, this problem is diffenent than the one addressed in this paper. Furthermore, the framework is only theoretical an based on complete systems.

\emph{Witnesses} have also been proposed in literature as an instrument to explain why a property is satisfied~\cite{Biere:1999:SMC:309847.309942,hong2002temporal,namjoshi2001certifying}.
Given an existential LTL formula and a model $M$, a witness is usually defined as a path that satisfies that formula.
This is different than the notion of topological proof  proposed in this work, where a proof is defined as a slice of the model $M$.

We are not aware of any work, except for \cite{Bernasconi2017},
that combines model checking and proofs in a multi-valued context.
Here the proposed proofs are verbose, obtained manually, and their effectiveness is not shown in practical applications.
This paper extends~\cite{Bernasconi2017} 
by defining topological proofs and model revisions, 
and by providing a working and practical environment.

 \section{Conclusions}
\label{sec:conclusions}

We have proposed \NAME , an integrated framework that allows a software designer to refine and revise her model proposal in a continuous verification setting. The framework, implemented in a tool for practical use, allows to specify partial models and properties to be verified. It checks these models against the requirements and provides a guide for the designer who wishes to preserve slices of her model that contribute to satisfy fundamental requirements while other components are modified.
For these purposes, the novel notion of topological proof has been formally and algorithmically described. This corresponds to a set of constraints that, if kept when changing the proposed model, ensure that the behavior of the model w.r.t. the property of interest is preserved.
\NAME\ was evaluated by showing the effectiveness of the \textsc{analysis} and \textsc{re-check} algorithms included in the framework. Results showed that proofs are in general easier to understand than the original models and thus \NAME\ can help the design process effectively.
Future work will consider supporting LTL thorough semantics in all phases of the framework and providing a deeper evaluation to show the speed-up offered by the \textsc{re-check} phase w.r.t. \textsc{analysis} re-execution.

\bibliographystyle{abbrv}

\end{document}